\documentclass[11pt]{article}
\usepackage{KLMM/klmm}

\title{A Finite Certificate for the Positive \(n=9\) Vasc Inequality}
\author{Dakai~Guo\thanks{State Key Laboratory of Mathematical Sciences, Academy of Mathematics and Systems Science, Chinese Academy of Sciences, and the School of Mathematics, University of Chinese Academy of Sciences, Beijing, People’s Republic of China. This work was supported by the Strategic Priority Research Program of Chinese Academy of Sciences under Grant XDA0480502 and XDA0480503. }
    \and Ruichen Qiu\samethanks
    \and Yichuan Cao\samethanks
    \and Ruyong Feng\samethanks
}
\date{\today}

\begin{document}
\maketitle

\begin{abstract}
We prove the positive-real \(n=9\) case of the Vasc cyclic inequality.  The
proof was obtained with human-guided assistance from the AI agent \textit{MechMath Agent Team}:
the human-readable part reduces the rational inequality to a homogeneous
polynomial inequality, fixes a cyclic maximum, and parametrizes each sorted
fixed-maximum cone by cumulative gaps; the finite part is a certificate
covering all \(8!=40320\) sorted cones.  \textit{MechMath Agent Team} generated the certificate
verification workflow through Python tool calls, including the case split,
verification programs, and terminal classifications.  The published certificate
has \(36815\) coefficient leaves, \(2236\) ordinary Polya multiplier leaves,
and \(1269\) AM-GM midpoint overlay leaves.  Human authors audited the
mathematical reductions and verification logic, and a separate artifact
contains the certificate, an independent verifier, and a from-source rebuild
route.
\end{abstract}

\section{Setup}
\label{sec:background-statement}

This first part fixes the convention for the cyclic expression, situates the positive-real \(n=9\) case, states the theorem, and outlines the proof route.

\subsection{Convention}
\label{subsec:convention}

For \(n\ge3\), the Vasc cyclic expression is
\[
C_n(x_1,\ldots,x_n)=
\sum_{i=1}^{n}
\frac{x_i-x_{i+1}}{x_{i+1}+x_{i+2}},
\]
with cyclic indexing. In the convention used in this paper, Vasc's conjecture asks for \(C_n(x)\ge0\) on the positive orthant. 
Closely related sources and discussions sometimes use the opposite sign and reversed cyclic order; the two forms are equivalent by reversing the variables and multiplying by \(-1\).
More explicitly, the public discussion \citep{MSEVasc} uses the expression
\[
F_n(u_1,\ldots,u_n)=
\sum_{i=1}^{n}\frac{u_{i+1}-u_{i+2}}{u_i+u_{i+1}}
\]
and asks for \(F_n(u)\le0\).
With cyclic indexing,
\[
C_n(x_1,\ldots,x_n)=-F_n(x_n,x_{n-1},\ldots,x_1).
\]
Consequently, under variable reversal, a proof of $F_n \le 0$ translates directly into a proof of $C_n \ge 0$, and a counterexample of the former immediately yields a counterexample of the latter.
The cleared polynomials in the two conventions are therefore not literally the same polynomial, but they are equivalent under this reversal and sign change, with multiplication only by positive denominator products.

\subsection{Related Work}
\label{subsec:related-work}

The Vasc family belongs to a broader research program on automated inequality proving, which has been extensively developed by Xia, Yang, and their collaborators.
A public discussion~\citep{YangXuezhi2009Inequalities, YangXuezhi2010TwentyTwo, MSEYangConj13} previously highlighted an issue regarding the thirteenth conjecture in Xuezhi Yang's list of twenty-two inequality conjectures.
Specifically, \citet{XiaYang} investigated this conjecture within the framework of successive difference substitution (SDS), building upon Yang's original difference-substitution method~\citep{Yang2006SDS}.
The SDS method establishes polynomial nonnegativity by iteratively applying structured substitutions until all remaining forms possess nonnegative coefficients.
The theoretical connection between Polya-type positivity certificates and SDS is detailed in~\citep{XuYaoSDS}.
However, within the Vasc family, the status of low-dimensional cases reported in the previous literature remains mixed. 
To clarify this, Table~\ref{tab:vasc-status} systematically categorizes these findings into three components: the documented base cases, the propagation principles used to extend counterexamples, and the cases addressed in this study.

\begin{table}[t]
    \centering
    \begin{tabular*}{\linewidth}{@{\hspace{5pt}\extracolsep{\fill}} l l p{0.65\linewidth} @{\hspace{5pt}}}
    \toprule
    \(n\) & Status & Source or method \\
    \midrule
    \(3,4,5,7\) & proved & SDS, as reported in \citep{XiaYang,MSEVasc} \\
    \(6\) & disproved & counterexample in \citep{XiaYang,MSEVasc} \\
    \(13\) & disproved & explicit boundary counterexample reported in \citep{MSEVascPropagation} \\
    \midrule
    \(n\mapsto n+2\) & - & \(F_{n+2}(u_1,\ldots,u_n,u_1,u_n)=F_n(u_1,\ldots,u_n)\), as recorded in \citep{MSEVascPropagation}\\
    \(2m,\ m\ge3\) & disproved & derived from the \(n=6\) counterexample \\
    \(2m+1,\ m\ge6\) & disproved & derived from the \(n=13\) counterexample \\
    \midrule
    \(9\) & \textbf{proved here} & \textbf{finite exact certificate with independent verification} \\
    \(11\) & open & partial positive result on one small cone~\citep{ZengEtAlCM2026Vasc}; full case not settled \\
    \bottomrule
    \end{tabular*}
    \caption{Status of Vasc inequalities. The first block lists the base cases recorded in earlier sources, and the second block lists the propagation principle used to extend counterexamples. The third block lists the cases addressed by this paper and the remaining open case. For consistency, all entries have been converted to the $C_n \ge 0$ convention.}
    \label{tab:vasc-status}
\end{table}

As delineated in the second block of Table~\ref{tab:vasc-status}, the propagation identity $n \mapsto n+2$ carries severe global implications. 
Specifically, the $n=6$ counterexample automatically rules out the validity of the inequality for all even dimensions $2m$ ($m \ge 3$), while the $n=13$ counterexample similarly disproves all odd dimensions $2m+1$ ($m \ge 6$). 
Although the counterexamples originally reported in the public discussion involve certain zero coordinates, small positive perturbations yield valid counterexamples within the strictly positive domain, as the denominators remain nonzero and the underlying expressions are continuous there~\citep{MSEVascPropagation}.


Accounting for these propagated invalid ranges, the only remaining open odd dimensions are $n=9$ and $n=11$. 
Regarding the $n=11$ case, a recent conference abstract reports that the inequality holds on a small cone in the first orthant~\citep{ZengEtAlCM2026Vasc}; however, this remains a partial result and leaves the full positive-real space unresolved. 
Both dimensions constitute the natural next frontiers for SDS-style reductions, yet direct scaling of lower-order patterns encounters a severe combinatorial explosion, with the initial ordering split requiring $8! = 40,320$ regions for $n=9$ and escalating to $10!$ regions for $n=11$. 
To address this substantial computational barrier, the primary contribution of this paper is the delivery of a finite exact certificate that formally proves the full positive-real case for $n=9$, thus successfully resolving the first open entry shown in the third block of Table~\ref{tab:vasc-status}.


\subsection{Statement}
\label{subsec:statement}

We now specialize the convention above to the case proved in this paper.
All indices below are interpreted cyclically modulo \(9\), so \(x_{9+i}=x_i\).

\begin{theorem}\label{thm:main}
Let \(x_1,\ldots,x_9\) be positive real numbers. Then
\[
\sum_{i=1}^9
\frac{x_i-x_{i+1}}{x_{i+1}+x_{i+2}}\ge 0.
\]
This theorem is the positive-real \(n=9\) Vasc inequality.
\end{theorem}

\subsection{Proof Strategy}
\label{subsec:proof-strategy}

The primary results of this paper were obtained using the AI agent \textit{MechMath Agent Team}~\footnote{\url{https://eonmath.github.io/mechmath}} under human guidance. 
\textit{MechMath Agent Team} is a large language model (LLM) driven agent designed to prove mathematical theorems expressed in natural language by planning algebraic reductions, invoking external tools, and iteratively verifying intermediate claims. 
In this study, the certificate verification phase was autonomously executed by \textit{MechMath Agent Team} via Python tool calls. 
Specifically, the agent authored the verification programs, generated the fixed-maximum and sorted-cone case splits, and organized the terminal cases into coefficient, Polya-multiplier, and AM-GM midpoint-circuit leaves. 
Then it repeatedly ran exact integer arithmetic checks until a terminal certificate was successfully verified for each of the 40,320 sorted cones. 
The human authors guided the overarching problem formulation and conducted final audits of the mathematical reductions, verification code, certificate formats, and reported verification results.

Our proof decouples human-verifiable algebraic reductions from the finite certificate check. 
Although certificate-based, it represents neither a raw computational transcript nor a complete proof-assistant formalization of the entire paper. 
Rather, the finite certificate is invoked at a single, well-defined juncture: immediately following the reduction to a fixed-maximum cone and the subsequent sorting of the remaining eight variables. 
For each of the 40,320 possible sorted cones, the certificate supplies a unique terminal row, which is then verified via three elementary soundness mechanisms: coefficientwise nonnegativity, Polya multiplication followed by division by a positive factor, and AM-GM midpoint circuit overlays.

\textbf{Proof sketch.}
The proof is structured into five distinct stages.
First, we clear the positive denominators to reduce the problem to analyzing a homogeneous polynomial $P_9$. 
Second, leveraging cyclic invariance, we restrict the verification of $P_9 \ge 0$ to tuples where the first coordinate is a maximum. 
Third, within each fixed-maximum cone, we sort the remaining variables and introduce cumulative gaps, which maps variable ties to zero gaps and transforms the domain into a nonnegative orthant with a positive total gap. 
Fourth, the finite certificate yields a terminal nonnegativity witness for every sorted cone. 
Fifth, the soundness of the three terminal leaf types establishes $P_9 \ge 0$, from which the original inequality is successfully recovered by reversing the denominator clearing.

\section{Reductions}
\label{sec:reductions}

This second part turns the original rational inequality into finitely many
orthant problems. We first clear positive denominators, then use cyclic
symmetry to fix a maximum coordinate, and finally sort the remaining
coordinates into \(8!\) cones, each parametrized by cumulative gaps.

\subsection{Cleared Polynomial}
\label{subsec:cleared-polynomial}

The first bridge is denominator clearing. Since the variables are strictly
positive, this step preserves the sign of the inequality.
For a positive 9-tuple \(x=(x_1,\ldots,x_9)\), put
\[
C_9(x)=\sum_{i=1}^9
\frac{x_i-x_{i+1}}{x_{i+1}+x_{i+2}}.
\]
Define
\[
D_9(x)=\prod_{i=1}^9 (x_{i+1}+x_{i+2})
\]
and
\[
P_9(x)=
\sum_{i=1}^9 (x_i-x_{i+1})
\prod_{\substack{1\le j\le 9\\ j\ne i}}
(x_{j+1}+x_{j+2}).
\]

\begin{lemma}[Denominator clearing]\label{lem:denom}
For every positive real 9-tuple \(x\),
\[
D_9(x)>0,\qquad P_9(x)=D_9(x)C_9(x).
\]
Consequently, \(P_9(x)\ge0\) implies \(C_9(x)\ge0\).
\end{lemma}

\begin{proof}
Let \(b_i=x_{i+1}+x_{i+2}\). Each \(b_i\) is positive, and therefore
\(D_9(x)=\prod_i b_i>0\). For each fixed \(i\),
\[
\frac{D_9(x)}{b_i}
=\prod_{\substack{1\le j\le 9\\ j\ne i}} b_j.
\]
Multiplying the finite sum defining \(C_9\) by \(D_9\) and substituting this
quotient identity gives
\[
D_9(x)C_9(x)
=\sum_{i=1}^9 (x_i-x_{i+1})
\prod_{\substack{1\le j\le 9\\ j\ne i}}
(x_{j+1}+x_{j+2})
=P_9(x).
\]
Since \(D_9(x)>0\), division by \(D_9(x)\) preserves the inequality direction.
\end{proof}

\subsection{Fixed-Maximum Reduction}
\label{subsec:fixed-maximum}

For \(k\in\{0,\ldots,8\}\), let \(R_kx\) be the cyclic rotation
\[
(R_kx)_i=x_{i+k}.
\]

\begin{lemma}[Cyclic maximum reduction]\label{lem:cyclicmax}
Assume that \(P_9(z)\ge0\) holds for every positive 9-tuple \(z\) satisfying
\[
z_i\le z_1\qquad (i=1,\ldots,9).
\]
Then \(P_9(x)\ge0\) holds for every positive 9-tuple \(x\).
\end{lemma}

\begin{proof}
The map \(i\mapsto i+k\) is a bijection of the cyclic index set. Reindexing
the defining finite sum and product for \(P_9\) gives
\[
P_9(R_kx)=P_9(x)
\]
for every \(k\). Now fix a positive tuple \(x\). Choose an index \(m\) for
which \(x_m\) is a maximum among the nine coordinates, and put \(k=m-1\). Then
\(z=R_kx\) is positive and satisfies \(z_i\le z_1\) for all \(i\). By the
fixed-maximum hypothesis, \(P_9(z)\ge0\). The cyclic invariance gives
\(P_9(x)=P_9(z)\), so \(P_9(x)\ge0\).
\end{proof}

\subsection{Sorted Cones and SDS Gap Coordinates}
\label{subsec:sorted-gaps}

After the maximum coordinate is fixed, the remaining variables are split
according to their weak order. The following coordinates make ties harmless:
they are represented by zero gaps.
Let \(\sigma=[p_1,\ldots,p_8]\) be a permutation of \(\{1,\ldots,8\}\). Set
\[
q_1=1,\qquad q_{m+1}=p_m+1\quad (m=1,\ldots,8),
\]
and define the fixed-maximum sorted cone
\[
K_{\sigma}=
\{z\in\R_{>0}^9:
z_1\ge z_{p_1+1}\ge\cdots\ge z_{p_8+1}>0\}.
\]
On this cone we use cumulative gap coordinates
\[
z_{q_m}=y_m+y_{m+1}+\cdots+y_9,\qquad m=1,\ldots,9,
\]
where
\[
y_1,\ldots,y_8\ge0,\qquad y_9>0.
\]
Let
\[
R_{\sigma}(y)=P_9(z(y))
\]
denote the pullback of \(P_9\) under this substitution.

\begin{lemma}[Sorted-gap bridge]\label{lem:gaps}
For each \(\sigma\), the above cumulative-gap map parametrizes \(K_{\sigma}\)
by the domain \(y_1,\ldots,y_8\ge0,\ y_9>0\). If \(z\in K_{\sigma}\), then the
associated gap vector satisfies
\[
y_1+\cdots+y_9=z_1>0
\]
and
\[
R_{\sigma}(y)=P_9(z).
\]
Weak ties in the sorted coordinates are represented by zero gaps.
\end{lemma}

\begin{proof}
If \(y_1,\ldots,y_8\ge0\) and \(y_9>0\), the cumulative sums satisfy
\[
y_m+\cdots+y_9\ge y_{m+1}+\cdots+y_9
\]
for \(m=1,\ldots,8\), and every coordinate is positive. Hence the constructed
\(z\) lies in \(K_{\sigma}\). Conversely, if \(z\in K_{\sigma}\), define
\[
y_m=z_{q_m}-z_{q_{m+1}}\quad (m=1,\ldots,8),\qquad y_9=z_{q_9}.
\]
The weak inequalities give \(y_m\ge0\) for \(m\le8\), and the last strict
inequality gives \(y_9>0\). Telescoping gives
\[
z_{q_m}=y_m+\cdots+y_9
\]
for every \(m\). Thus the two constructions are inverse, the sum of all gaps
is \(z_{q_1}=z_1>0\), and \(R_{\sigma}(y)=P_9(z)\) by definition. Adjacent
equalities in the sorted chain are exactly the zero adjacent gaps.
\end{proof}

\section{Certificate Verification}
\label{sec:certificate-verification}
This final stage represents the sole juncture at which the finite certificate is integrated into the proof. 
The certificate maps a single terminal row to each sorted cone, while the subsequent lemmas establish the mathematical soundness of the three permissible terminal row types.
The underlying trust model is deliberately narrow in scope. 
Specifically, this implementation does not constitute a formalization within an interactive proof assistant such as Rcoq, Lean, or Isabelle. 
Furthermore, producer programs are excluded from the core proof, serving instead as a reproducibility pipeline for regenerating the certificate. 
The proof incorporates the certificate only after it has been validated by an independent, minimal verifier. 
This verifier relies exclusively on the Python standard library, ensuring complete isolation from any producer modules. 
Using exact integer arithmetic, it recomputes the fixed-maximum stream, pullbacks, polynomial hashes, Polya products, and AM-GM residual checks directly from the stored certificate data.



\subsection{The Fixed-Maximum Certificate}
\label{subsec:fixed-maximum-certificate}

The only computational input is the following finite certificate. It concerns
the lexicographic stream of all permutations of \(\{1,\ldots,8\}\), indexed by
the full half-open interval \([0,40320)\). The independent verifier checks
that this stream is bound to the displayed polynomial \(P_9\) and to the
fixed-maximum substitution convention used above.  The certificate data and
verifier are distributed as the standalone artifact repository
\begin{center}
\url{https://github.com/EonMath/Vasc-n9-cert/releases/tag/v1.0.0}.
\end{center}
Appendix~\ref{app:certificate-artifact} gives the verification and
reproduction details.

\begin{proposition}[Full fixed-maximum certificate]\label{prop:certificate}
For each of the \(40320\) roots in the fixed-maximum stream, the certificate
contains exactly one final row. There are no missing roots, no duplicate roots,
no deferred rows, and no unresolved final rows. Each final row has exactly one
of the following three terminal types:
\begin{enumerate}
\item coefficient leaf;
\item ordinary Polya multiplier leaf;
\item AM-GM midpoint overlay leaf.
\end{enumerate}
The terminal counts are
\[
36815,\qquad 2236,\qquad 1269,
\]
respectively, and these counts sum to \(40320\).
\end{proposition}

This proposition is the finite machine-checked input to the proof.  The
surrounding lemmas below are ordinary mathematical arguments showing that any
row passing one of the three terminal checks implies nonnegativity on its
corresponding sorted cone.

The terminal labels have the following mathematical meanings. Let
\[
L(y)=y_1+\cdots+y_9.
\]
Here \(\N=\{0,1,2,\ldots\}\). A coefficient leaf asserts that
\(R_{\sigma}\) is coefficientwise nonnegative. An ordinary Polya multiplier
leaf asserts that, for some \(k\in\N\), the polynomial
\[
Q=L^kR_{\sigma}
\]
is coefficientwise nonnegative. An AM-GM midpoint overlay leaf asserts that,
for some \(k\in\N\),
\[
Q=L^kR_{\sigma}=S+C_1+\cdots+C_N,
\]
where \(S\) is coefficientwise nonnegative and each circuit \(C_\nu\) has the
form
\[
C_\nu(y)
=m_\nu y^{\alpha_\nu}+m_\nu y^{\gamma_\nu}
-2m_\nu y^{\beta_\nu},
\]
with \(m_\nu>0\), exponent vectors
\(\alpha_\nu,\beta_\nu,\gamma_\nu\in\N^9\), and
\[
\alpha_\nu+\gamma_\nu=2\beta_\nu
\]
coordinatewise. In the overlay case, the verifier checks matching pullback
hashes, nonnegative residual data, positive circuit amounts, and the midpoint
exponent identities for the circuits used in the final row.

\subsection{Leaf Soundness}
\label{subsec:leaf-soundness}

The certificate has three terminal labels. Coefficient and Polya leaves are
immediate from coefficientwise nonnegativity after division by a positive
power of \(L\); the only extra ingredient for overlay leaves is the following
AM-GM midpoint circuit.

\begin{lemma}[Midpoint circuit soundness]\label{lem:circuit}
Let \(m>0\), let \(\alpha,\beta,\gamma\in\N^9\), and suppose
\(\alpha+\gamma=2\beta\) coordinatewise. Then, for every
\(y\in\R_{\ge0}^9\),
\[
m y^\alpha+m y^\gamma-2m y^\beta\ge0.
\]
\end{lemma}

\begin{proof}
For each coordinate choose nonnegative integer vectors \(p,a,b\) such that
\[
\alpha=p+2a,\qquad \gamma=p+2b,\qquad \beta=p+a+b.
\]
Indeed, in a coordinate with \(\alpha_i\le\gamma_i\), set
\[
p_i=\alpha_i,\qquad a_i=0,\qquad b_i=\frac{\gamma_i-\alpha_i}{2};
\]
the midpoint identity gives \(\gamma_i-\alpha_i=2(\beta_i-\alpha_i)\), so
\(b_i\in\N\). The case \(\gamma_i<\alpha_i\) is symmetric, with
\[
p_i=\gamma_i,\qquad a_i=\frac{\alpha_i-\gamma_i}{2},\qquad b_i=0.
\]
With
\[
A=y^a,\qquad B=y^b,\qquad M=y^p,
\]
the usual laws of natural powers on the nonnegative orthant, with \(0^0=1\),
give
\[
y^\alpha=MA^2,\qquad y^\gamma=MB^2,\qquad y^\beta=MAB.
\]
Therefore
\[
m y^\alpha+m y^\gamma-2m y^\beta
=mM(A-B)^2\ge0,
\]
because \(m>0\), \(M\ge0\), and squares are nonnegative.
\end{proof}

\begin{lemma}[Terminal leaf soundness]\label{lem:leafsound}
Let \(y_1,\ldots,y_9\ge0\) and \(L(y)>0\). For a fixed \(\sigma\), any one of
the three terminal certificates in Proposition~\ref{prop:certificate} implies
\[
R_{\sigma}(y)\ge0.
\]
\end{lemma}

\begin{proof}
If the row is a coefficient leaf, then
\[
R_{\sigma}=\sum_a c_a y^a
\]
with finitely many terms and all \(c_a\ge0\). Since every monomial \(y^a\) is
nonnegative on \(\R_{\ge0}^9\), \(R_{\sigma}(y)\ge0\).

If the row is an ordinary Polya multiplier leaf, then
\[
Q=L^kR_{\sigma}
\]
for some \(k\in\N\), and \(Q\) is coefficientwise nonnegative. Hence
\(Q(y)\ge0\). Since \(L(y)>0\), also \(L(y)^k>0\), including \(k=0\). Thus
\[
R_{\sigma}(y)=\frac{Q(y)}{L(y)^k}\ge0.
\]

If the row is an AM-GM midpoint overlay leaf, then
\[
Q=L^kR_{\sigma}=S+C_1+\cdots+C_N,
\]
where \(S\) is coefficientwise nonnegative and each \(C_\nu\) satisfies the
hypotheses of Lemma~\ref{lem:circuit}. Therefore \(S(y)\ge0\) and
\(C_\nu(y)\ge0\) for every \(\nu\), so \(Q(y)\ge0\). The same positive
division by \(L(y)^k\) gives \(R_{\sigma}(y)\ge0\).
\end{proof}

\subsection{Proof of Theorem~\ref{thm:main}}
\label{subsec:proof-main}

We now assemble the reductions and the terminal certificate soundness.

\begin{proof}
By Lemmas~\ref{lem:denom} and~\ref{lem:cyclicmax}, it is enough to prove
\[
P_9(z)\ge0
\]
for every positive tuple \(z\) with \(z_i\le z_1\) for all \(i\).

Fix such a tuple \(z\). Sort \(z_2,\ldots,z_9\) in weakly decreasing order.
Thus there is a permutation \(\sigma=[p_1,\ldots,p_8]\) of
\(\{1,\ldots,8\}\) such that
\[
z_{p_1+1}\ge z_{p_2+1}\ge\cdots\ge z_{p_8+1}.
\]
If ties occur, choose the lexicographically least sorting permutation. This
tie rule only selects one compatible row of the finite certificate and imposes
no strict inequality. Since \(z_1\) is a maximum and all coordinates are
positive,
\[
z_1\ge z_{p_1+1}\ge\cdots\ge z_{p_8+1}>0,
\]
so \(z\in K_{\sigma}\).

Lemma~\ref{lem:gaps} supplies a gap vector \(y\) with
\[
y_1,\ldots,y_8\ge0,\qquad y_9>0,\qquad L(y)=y_1+\cdots+y_9>0,
\]
and
\[
R_{\sigma}(y)=P_9(z).
\]
The selected permutation has a unique index in the full lexicographic stream
\([0,40320)\). By Proposition~\ref{prop:certificate}, its final row has one
of the three terminal types. Lemma~\ref{lem:leafsound} applies in each case
and gives
\[
R_{\sigma}(y)\ge0.
\]
Using \(R_{\sigma}(y)=P_9(z)\), we obtain \(P_9(z)\ge0\).

The fixed-maximum tuple \(z\) was arbitrary. Lemma~\ref{lem:cyclicmax}
therefore gives \(P_9(x)\ge0\) for every positive real 9-tuple \(x\). Applying
Lemma~\ref{lem:denom} to the same tuple gives
\[
C_9(x)=\frac{P_9(x)}{D_9(x)}\ge0,
\]
because \(D_9(x)>0\). Since \(x\) was arbitrary, the positive-real \(n=9\)
Vasc inequality follows.
\end{proof}

\section{Limitations}
\label{subsec:limitations}

This paper proves only Theorem~\ref{thm:main}. It does not include a limiting
argument for zero-coordinate boundary tuples, and it does not address the
\(n=11\) case. The finite certificate is a full certificate only for the
fixed-maximum \(n=9\) stream after the sorted-gap substitution. It is connected
to the original rational inequality by the algebraic bridges proved in this
paper: denominator clearing, cyclic maximum reduction, sorted-gap
parametrization, and terminal leaf soundness.

\section*{Acknowledgements}

We thank Wensheng Yu and Guowei Dou of Beijing University of Posts and
Telecommunications for informing us that a proposed \(n=9\) counterexample to
the positive-real Vasc inequality does not in fact violate the inequality.
Their observation led us to reexamine the \(n=9\) case and motivated the
present certificate proof.  We also thank Zhenbing Zeng and coauthors for
their conference abstract on the \(n=9\) and \(n=11\) Vasc cases
\citep{ZengEtAlCM2026Vasc}, which helped draw attention to these cases.

\bibliography{ref}
\bibliographystyle{KLMM/klmm}
\clearpage

\appendix

\section{Certificate Artifact and Reproducibility Details}
\label{app:certificate-artifact}

The finite certificate used in Proposition~\ref{prop:certificate} is the
fixed-maximum \(n=9\) packet named
\[
\texttt{prefix\_0000000\_0040320\_with\_hardroots}.
\]
The independent verifier checks that its scope is exactly the full half-open
range \([0,40320)\) in the fixed-maximum stream, with one final row for every
root index, no gaps, no overlaps, no duplicate roots, no deferred rows, and no
unresolved final rows. The terminal status counts are \(36815\), \(2236\),
and \(1269\), as stated above.

The artifact repository is named \texttt{vasc-n9-cert}.  It contains:
\begin{itemize}
\item \texttt{certificates/}, containing the base Polya batches, the hardroot
overlay packets, and the final packet;
\item \texttt{tools/verify\_n9\_certificate\_minimal.py}, an independent
checker for the published certificate;
\item \texttt{tools/rebuild\_n9\_certificate\_from\_sources.py}, an optional
driver for regenerating the certificate files;
\item \texttt{producer\_sources/}, the producer source snapshot used by the
rebuild driver;
\item \texttt{rebuild/n9\_rebuild\_plan.json}, the frozen rebuild schedule;
\item \texttt{SHA256SUMS}, a checksum list for the packaged files;
\item \texttt{README.md} and \texttt{REPRODUCE.md}, which give reviewer-facing
verification and rebuild instructions.
\end{itemize}

There are two reproducibility layers.  The independent verification layer is
the proof-relevant one: it checks the deposited certificate using a small
checker independent of the producer programs.  The rebuild layer is optional:
it regenerates the same certificate files from the producer source snapshot and
then compares the regenerated files against \texttt{SHA256SUMS}.

\subsection{Independent Verification}
\label{app:independent-verification}

The recommended reviewer check is the independent verifier.  From the artifact
root, it is run by
\begin{center}
\texttt{python3 tools/verify\_n9\_certificate\_minimal.py}.
\end{center}
The checker uses only the Python standard library and does not import the
producer modules.  It recomputes the lexicographic permutation for each root,
the direct pullback of \(P_9\), the relevant polynomial hashes, the ordinary
Polya products, and the AM-GM midpoint residuals using exact integer
arithmetic.  A successful full run ends with \texttt{PASS} and the counts
\begin{center}
\begin{tabular}{rl}
\texttt{coefficient\_leaf\_count} & \(36815\),\\
\texttt{polya\_leaf\_count} & \(2236\),\\
\texttt{amgm\_midpoint\_overlay\_leaf\_count} & \(1269\),\\
\texttt{unresolved\_count} & \(0\).
\end{tabular}
\end{center}
The full run is single-threaded and long-running; on the preparation machine it
was estimated to take about 24--35 hours, and the largest AM-GM overlay roots
may use a few GB of memory.  The artifact also provides
two smoke checks:
\begin{center}
\texttt{python3 tools/verify\_n9\_certificate\_minimal.py --limit 32 --quiet}
\end{center}
and
\begin{center}
\texttt{python3 tools/verify\_n9\_certificate\_minimal.py --root-index 5952 --quiet}.
\end{center}
The first checks the first 32 roots, and the second checks an AM-GM overlay
root.
For archived verification runs, the verifier accepts
\texttt{--run-log <path>} and writes a structured JSON run record together with
a \texttt{.sha256} sidecar for that record.

File integrity can be checked independently by
\begin{center}
\texttt{sha256sum -c SHA256SUMS}.
\end{center}
This checksum command is not the mathematical verification; it only checks
that the packaged files match the archived artifact.

\subsection{From-Source Rebuild}
\label{app:from-source-rebuild}

The optional rebuild route regenerates the certificate files from the included
producer source snapshot.  The plan summary is printed by
\begin{center}
\texttt{python3 tools/rebuild\_n9\_certificate\_from\_sources.py --mode plan}.
\end{center}
A smoke rebuild in a fresh workspace is run by
\begin{center}
\begin{tabular}{l}
\texttt{python3 tools/rebuild\_n9\_certificate\_from\_sources.py --mode smoke}\\
\texttt{--build-workspace /tmp/vasc\_n9\_rebuild\_smoke}\\
\texttt{--run-log logs/rebuild\_smoke.json}.
\end{tabular}
\end{center}
The full rebuild regenerates all \(633\) raw Polya batches, all \(281\)
hardroot overlay packets, and the final packet, then compares the regenerated
certificate files against \texttt{SHA256SUMS}.  This gives a source-level
route for reproducing the deposited certificate, separate from the independent
mathematical verifier.
The rebuild driver records its subprocess commands, hash-comparison summaries,
and final status in the requested run log, again with a \texttt{.sha256}
sidecar.

\end{document}